\documentclass[a4paper,12pt]{article}

\usepackage{amssymb}
\usepackage{amsthm}
\usepackage[ruled,vlined,linesnumbered]{algorithm2e}
\usepackage{array}
\usepackage{diagbox}
\usepackage{threeparttable}
\usepackage{graphicx}
\usepackage{subfigure}
\usepackage{amsmath}
\usepackage{multirow}
\usepackage{color}
\usepackage{cite}
\usepackage[noend]{algpseudocode}
\usepackage{natbib}
\usepackage{booktabs}
\usepackage{geometry}
\usepackage{hyperref}
\newtheorem{definition}{Definition}

\newtheorem{theorem}{Theorem}

\newcommand\independent{\protect\mathpalette{\protect\independenT}{\perp}}
\def\independenT#1#2{\mathrel{\rlap{$#1#2$}\mkern2mu{#1#2}}}

\DeclareMathOperator*{\argmin}{argmin}
\DeclareMathOperator*{\argmax}{argmax}

\begin{document}


\title{\bf Certifiably Optimal Sparse Sufficient Dimension Reduction}

\author{Lei Yan and Xin Chen\thanks{  Corresponding author} \\
Department of Statistics and Data Science, Southern University of\\
Science and Technology}

\date{\today}
\maketitle

\bigskip
\begin{abstract}
Sufficient dimension reduction (SDR) is a popular tool in regression analysis, which replaces the original predictors with a minimal set of their linear combinations. However, the estimated linear combinations generally contain all original predictors, which brings difficulties in interpreting the results, especially when the number of predictors is large. In this paper, we propose a customized branch and bound algorithm, optimal sparse generalized eigenvalue problem (Optimal SGEP), which combines a SGEP formulation of many SDR methods and efficient and accurate bounds allowing the algorithm to converge quickly. Optimal SGEP exactly solves the underlying non-convex optimization problem and thus produces certifiably optimal solutions. We demonstrate the effectiveness of the proposed algorithm through simulation studies. 
\end{abstract}


\section{Introduction}  
In regression analysis, sufficient dimension reduction (SDR) (\citet{li1991sliced,cook1994interpretation,cook1998regression}) is a statistical technique developed to reduce the dimension of covariate without loss of regression information. For regression problems with a response variable $y \in \mathbb{R}$ and a predictor vector $\mathbf{x} \in \mathbb{R}^p$, a dimension reduction subspace is any subspace $\mathcal{S} \subseteq \mathbb{R}^p$ such that $y$ is independent of $\mathbf{x}$ given the orthogonal projection of $\mathbf{x}$ on to $\mathcal{S}$:
\begin{equation}
    y \independent \mathbf{x} | P_{\mathcal{S}}\mathbf{x},
\end{equation}
where $\independent$ represents independence. The intersection of all dimension reduction subspaces, denoted by $\mathcal{S}_{ y|\mathbf{x}}$, is called central subspace. Under mild conditions, the central subspace exists and is the unique minimum dimension reduction subspace (\citet{cook1998regression,yin2008successive}). The structural dimension $d = \text{dim}(\mathcal{S}_{y|\mathbf{x}})$ is generally far less than the original dimension $p$, which implies we can greatly reduce the dimensionality of the covariate space. The primary goal of SDR is to estimate such central subspaces.

Many methods for estimating the central subspace have been proposed since the introduction of sliced inverse regression (SIR) (\citet{li1991sliced}) and sliced average variance estimation (SAVE) (\citet{cook1991discussion}). \citet{li1992principal} and \citet{cook1998principal} proposed principal Hessian directions (PHD). \citet{li2007directional} proposed directional regression (DR) which combines advantages of SIR and SAVE. \citet{cook2007fisher} discussed principal fitted components (PFC), a likelihood based method which is efficient in estimating the central subspace. \citet{ma2012semiparametric,ma2013efficient} adopted a semiparametric approach and developed an efficient estimation procedure for estimating the central subspace.

One major drawback of the aforementioned SDR methods is the estimated sufficient dimension reduction directions contain all $p$ predictors, which brings difficulty in interpreting these directions and identifying important variables. Many attempts have been made to tackle this problem. \citet{li2005model} proposed a model-free variable selection approach. \citet{ni2005note} employed LASSO penalty (\citet{tibshirani1996regression}) to obtain a shrinkage SIR estimator. \citet{li2007sparse} studied a unified framework called sparse SDR, which can apply to most SDR methods. \citet{chen2010coordinate} proposed coordinate-independent sparse estimation (CISE) which simultaneously achieves dimension reduction and variable selection. \citet{tan2018convex} formulated the sparse SIR as a convex optimization problem and solved it using linearized alternating direction method of multipliers. Almost at the same time, \citet{tan2018sparse} proposed a more general algorithm called truncated Rayleigh flow (Rifle).

In all the methods introduced above, the algorithms of \citet{li2007sparse}, \citet{chen2010coordinate} and \citet{tan2018sparse} perform variable selection via solving a sparse generalized eigenvalue problem (SGEP) and thus they are suitable for many SDR methods. The relationship between these SDR methods and a generalized eigenvalue problem is reviewed in Section \ref{sec_prob_formulation}. However, their methods either impose heuristic $\ell_1$ penalty, not the intrinsic $\ell_0$ constraint, or solve non-convex problems inexactly and a globally optimal solution is usually not guaranteed.

SGEP with a $\ell_0$ constraint is a mixed integer optimization (MIO) problem, which belongs to a class of NP-hard problems (\citet{moghaddam2006generalized}). Traditionally, NP-hard problems have been deemed intractable and methods that solve problems with practical size in reasonable time were unavailable. With the rapid development of MIO solvers (\cite{bixby2012brief}), such CPLEX and Gurobi, some researchers have attempted to apply MIO to problems in Statistics (\citet{carrizosa2014rs,asteris2015sparse,bertsimas2016best,bertsimas2017optimal,berk2019certifiably,bertsimas2020sparse}). 

Motivated by those work, we fill in the blank of solving sparse SDR with MIO in this paper. We propose a customized branch and bound algorithm that enable us to solve non-convex SGEP with $\ell_0$ a constraint to certifiable optimality. In this paper, we assume that $n > p$ and the structural dimension $d$ is known. 

\subsection{Notation}
Boldface lowercase letters like $\mathbf{v}$, uppercase letters like $A$ and $B$, and script capital letter like $\mathcal{I}$ are used to represent vectors, matrices and sets, respectively. For a set $\mathcal{I}$, we use $|\mathcal{I}|$ to denote its cardinality. Let $\lambda_{\max}(A)$ and $\lambda_{\min}(A)$ be the largest and smallest eigenvalues correspondingly, and $\lambda_{\max}(A, B)$ be the largest generalized eigenvalue of a matrix pair $(A, B)$. We denote the nuclear norm of A by $||A||_* = tr(A)$, the spectral norm of A by $||A||_2=\sqrt{\lambda_{\max}(A^{\intercal}A)}$, and the Frobenius norm of A by $||A||_F = \sqrt{tr(A^{\intercal}A)} $. For a matrix $A$, $A_{ij}$ and $A_{i\cdot}$ represent the element in row $i$ and column $j$ and the $i$-th row, respectively. For $\mathcal{I} \in \{1, 2, \dots, p \}$ and $\mathbf{v} \in \mathbb{R}^p$, let $\mathbf{v}_{\mathcal{I}}$ be the subvector of $\mathbf{v}$ where elements are restricted to the set $\mathcal{I}$. Finally, let $\mathbf{e}_i$ be a vector whose elements are all zero, except the $i$-th element that equals 1, and $\mathbb{S}^p_{+}$ be the set of semi-positive matrices.

\subsection{Organization}
This paper is organized as follows. In Section \ref{sec_prob_formulation}, we review how to write SDR methods as a generalized eigenvalue problem and formulate the SGEP into a MIO form. We describe a branch and bound algorithm in Section \ref{sec_B&B} which solves the MIO formulation to certifiable optimality. Our branch and bound algorithm benefits greatly from the bound functions. Thus, in Section \ref{sec_upper_lower_bounds}, we derive three upper bounds and two lower bounds for the sub-problems of the MIO problem. In Section \ref{node_selection_branch}, we provide some suggestions about methods of branching and selecting nodes, as well as about details on computing lower bounds, through numerical experiments. The convergence property of our proposed algorithm is established in Section \ref{convergence}. To select the tinning parameter $k$, AIC and BIC criteria are introduced in Section \ref{sec_tuning}. In Section \ref{simulation_section}, we report results of simulation studies, which demonstrate the superior performance of our proposed algorithm. Finally, we gives some concluding remarks in Section \ref{sec_conclusion}.

\section{Methodology}  
\subsection{Problem formulation} \label{sec_prob_formulation}
It is shown that many sufficient dimension reduction methods can be formulated as a generalized eigenvalue problem (\citet{li2007sparse,chen2010coordinate}) of the form
\begin{equation}
  \begin{aligned}
  &A\mathbf{v}_i = \lambda_i B \mathbf{v}_i, \quad i = 1, 2, \dots, p, \\
  &\text{s.t. } \mathbf{v}_i^{\intercal}B\mathbf{v}_j = \begin{cases}
  1 & \text{if } i = j, \\
  0 & \text{if } i \neq j,
  \end{cases}
  \end{aligned}
  \label{gep}
  \end{equation}
where $A$ is a method specific positive semi-definite matrix; $B$ is a symmetric positive definite matrix, which is the sample covariance matrix in most methods; $\lambda_1 \geq \dots \geq \lambda_p $ are generalized eigenvalues; and $\mathbf{v}_1, \dots, \mathbf{v}_p$ are corresponding generalized eigenvectors. Under some mild conditions on the marginal distribution of $\mathbf{x}$, it can be shown that the first $d$ eigenvectors span the central subspace. Table \ref{gepsummrize} summarizes some commonly used SDR methods and the corresponding matrices A and B.

\begin{center}
\begin{table}[htbp] 
 \caption{\label{gepsummrize}The GEP formulation of some SDR methods} 
 \begin{tabular}{lcl} 
  \toprule 
  Method & A & B \\ 
  \midrule 
 SIR & $\text{cov}[E\{\textbf{x} - E(\textbf{x})\} \mid y ]$ & $\Sigma_\mathbf{x} $ \\ 
 PFC & $\Sigma_{fit}$ & $\Sigma_{\mathbf{x}}$ \\
 SAVE & $\Sigma_\mathbf{x}^{1 / 2} E\left[\{I-\operatorname{cov}(\textbf{z} \mid y)\}^{2}\right] \Sigma_\mathbf{x}^{1 / 2},$ where $\textbf{z}=\Sigma_\mathbf{x}^{-1 / 2}\{\mathbf{x}-E(\mathbf{x})\}$  & $\Sigma_\mathbf{x} $ \\ 
 PHD(y-based) & $\Sigma_{\mathbf{x}}^{1 / 2} \Sigma_{y\mathbf{z z}} \Sigma_{y\mathbf{ z z}} \Sigma_{\mathbf{x}}^{1 / 2},$ where $\Sigma_{y\mathbf{z z}}=E\left[\{y-E(y)\} \textbf{z} \textbf{z}^{\top}\right]$ & $\Sigma_\mathbf{x} $ \\ 
 PHD(r-basied) & $\Sigma_{\mathbf{x}}^{1 / 2} \Sigma_{\mathbf{r z z}} \Sigma_{\mathbf{r z z}} \Sigma_{\mathbf{x}}^{1 / 2} , $ where $\Sigma_{\mathbf{r z z}}=E\left[\left\{y-E(y)-E\left(y \mathbf{z}^{\intercal}\right) \mathbf{z}\right\} \mathbf{z} \mathbf{z}^{\intercal}\right]$ & $\Sigma_{\mathbf{x}} $ \\
 DR  & $\Sigma_{\mathbf{x}}^{1 / 2}\left\{2 E\left[E^{2}\left(\mathbf{z} \mathbf{z}^{\intercal} \mid y\right)\right]+2 E^{2}\left[E(\mathbf{z} \mid y) E\left(\mathbf{z}^{\intercal} \mid y\right)\right]\right.$ & \\
& $\left.+2 E[E(\mathbf{z} \mid y) E(\mathbf{z} \mid y)] E\left[E(\mathbf{z} \mid y) E\left(\mathbf{z}^{\intercal} \mid y\right)\right]-2 \mathbf{I}_{p}\right\} \Sigma_{\mathbf{x}}^{1 / 2}$ & $\Sigma_{\mathbf{x}}$ \\
  \bottomrule 
 \end{tabular} 
\end{table}
 \end{center}

Note that the first eigenvector $\mathbf{v}_1$ can be characterized as 
\begin{equation}
    \max_{\mathbf{v}} \mathbf{v}^{\intercal}A\mathbf{v}, \quad \text{s.t. } \mathbf{v}^{\intercal}B\mathbf{v} = 1.
    \label{eq_gep}
\end{equation}
Denote the solution of \eqref{eq_gep} by $\hat{\mathbf{v}}_1$. To obtain the second eigenvector, we first project the matrix $A$ into the subspace perpendicular to $\hat{\mathbf{v}}_1$:
$$
A_2 = \left(I - \frac{B\hat{\mathbf{v}}_1(B\hat{\mathbf{v}}_1)^{T}}{||B\hat{\mathbf{v}}_1||^2_2} \right)A\left(I - \frac{B\hat{\mathbf{v}}_1(B\hat{\mathbf{v}}_1)^{T}}{||B\hat{\mathbf{v}}_1||^2_2} \right),
$$
and then solve the GEP with matrix $A_2$. Subsequent eigenvectors can be found in the same way. When $p$ is large, we assume the sufficient directions or eigenvectors are sparse. The first sparse eigenvector can be obtained by solving the sparse generalized eigenvalue problem (SGEP)
\begin{equation}
     \max_{\mathbf{v}} \mathbf{v}^{\intercal}A\mathbf{v}, \quad \text{s.t. } \mathbf{v}^{\intercal}B\mathbf{v} = 1, ||\mathbf{v}||_0 \leq k,
     \label{sgep}
\end{equation}
where the $\ell_0$ norm constrains the number of nonzero elements. 

Now we rewrite problem \eqref{sgep} into an explicit mixed integer optimization form called SGEP-MIO

\begin{equation}
    \begin{aligned}
        \max_{\mathbf{v, w}} \quad & \mathbf{v}^{\intercal}A\mathbf{v}  \\
         \text{s.t. } \quad & \sum_{i=1}^p \lambda_i(\mathbf{s}_i^{\intercal}\mathbf{v})^2 = 1   \\
         & -w_iM_i \leq v_i \leq w_iM_i, \quad i = 1, \dots, p \\
         & \sum_{i=1}^p w_i=k  \\
         & \mathbf{v} \in [-M_{\text{max}}, M_{\text{max}}]^p \\
         & \mathbf{w} \in \{ 0, 1 \}^p, \\
    \end{aligned}
    \label{sgep_mio}
\end{equation}

where $\lambda_i$ is the $i$-th largest eigenvalue of B and $\mathbf{s}_i$ is the corresponding eigenvector; $M_i > 0$ is a large number related to $B$ and $M_{\text{max}} = \max_{0 \leq i \leq p} \{ M_i \}$; $\mathbf{w}$ is the support vector of $\mathbf{v}$ which controls the sparsity. 

\subsection{A branch and bound algorithm}\label{sec_B&B}
In general, the SGEP-MIO is non-convex, NP-hard and therefore intractable (\citet{moghaddam2006generalized}). For the special case $A = \Sigma_x, B = I$, the existing mixed integer optimization solvers was unable to solve problems with $n=13, k=10$ in a hour (\citet{berk2019certifiably}). In this section, we'll develop an efficient branch and bound algorithm called Optimal SGEP which exactly solves the non-convex problem \eqref{sgep_mio}.  

The Optimal SGEP constructs a tree of subproblems which are formed by determining some elements of the support vector $\mathbf{w}$:

\begin{equation}
    \begin{aligned}
        \max_{\mathbf{v, w}} \quad & \mathbf{v}^{\intercal}A\mathbf{v}  \\
         \text{s.t. } \quad & \sum_{i=1}^p \lambda_i(\mathbf{s}_i^{\intercal}\mathbf{v})^2 = 1   \\
         & -w_iM_i \leq v_i \leq w_iM_i, \quad i = 1, \dots, p \\
         & \sum_{i=1}^p w_i=k  \\
         & \mathbf{v} \in [-M_{\text{max}}, M_{\text{max}}]^p \\
         & \mathbf{l} \leq \mathbf{w} \leq \mathbf{u}, \quad \mathbf{l, w, u} \in \{ 0, 1 \}^p, \\
    \end{aligned}
    \label{sgep_pds}
\end{equation}
where $\mathbf{l, u}$ are lower and upper bounds on $\mathbf{w}$. Each node of the tree is identified by a pair $(\mathbf{l, u})$. Beginning from the original problem \eqref{sgep_mio} which corresponds to $(\mathbf{l} = \{ 0 \}^p, \mathbf{u} = \{ 1 \}^p)$, the entire tree consists of $2^{p+1} - 1$ nodes. If we can solve all the subproblems, the original SGEP-MIO is solved. However, enumerating the entire tree is inefficient and impractical. Fortunately, we can avoid exploring some unnecessary nodes or subtrees. 

Denote the feasible set of problem \eqref{sgep_pds} by $\mathcal{V}(\mathbf{l, u}, k, B)$. In the case of $\sum_{i=1}^p l_i \leq k$ and $\sum_{i=1}^p u_i \geq k$, $\mathcal{V}(\mathbf{l, u}, k, B)$ is an non-empty set. In particular, if $\sum_{i=1}^p l_i = \sum_{i=1}^p u_i = k$, the support $\mathbf{w}$ is fixed. By eliminating rows and columns of $A, B$ that correspond to $w_i = 0$, this subproblem becomes a regular generalized eigenvalue problem which can be solved efficiently. Note that if we continue branching on this kind of node to produce new nodes, the feasible sets will be empty. Thus we call nodes with $\sum_{i=1}^p l_i \geq 
k, \sum_{i=1}^p u_i \leq k$ terminal nodes. By testing whether a node is terminal, we can reduce the number of nodes to explore. Furthermore, we only need to investigate nodes with possible superior solutions. In other words, we only need to explore nodes with upper bounds greater than the current estimated optimal solution. Actually, the performance of branch and bound algorithm relies heavily on two functions that calculate the upper and lower bounds of a subproblem
\begin{equation}
\begin{aligned}
upper(\mathbf{l, u}, k, B) &\geq \max \mathbf{v}^{\intercal}A\mathbf{v}, \quad \mathbf{v} \in \mathcal{V}(\mathbf{l, u}, k, B), \\
lower(\mathbf{l, u}, k, B) &= \mathbf{v}^{\intercal}A\mathbf{v}, \quad \mathbf{v} \in \mathcal{V}(\mathbf{l, u}, k, B).
\end{aligned}
\label{ub_lb_func}
\end{equation}
For now, we assume these bounds can be computed efficiently and we'll derive them in Section \ref{sec_upper_lower_bounds}.

The main crux of the Optimal SGEP algorithm is as follows. Starting from the root $(\mathbf{l, u}) = \left( \{ 0 \}^p, \{ 1 \}^p\right)$, we maintain a set of unexplored nodes. In each iteration, one node is selected and two new nodes are created by branching on a dimension. If the new nodes are terminal, the subproblem becomes a regular GEP and the exact solution is returned. Otherwise, the lower and upper bounds of the nodes are computed using corresponding functions. We update the optimal solution if it is superior to the best solution found so far. Some nodes are deleted since their upper bounds are smaller than the best feasible solution. The above steps are repeated until the difference between the global upper bound and the optimal solution is within $\epsilon$.

\newpage

\begin{algorithm}[H]
\SetAlgoLined
\SetKwInOut{Input}{Input}
\SetKwInOut{Output}{Output} 
\Input{matrices $A, B$, cardinality $k$, tolerance $\epsilon$}
\Output{sparse vector $\hat{\mathbf{v}}$ with $||\hat{\mathbf{v}}||_0 \leq k$ and $||\hat{\mathbf{v}}^{\intercal}B\hat{\mathbf{v}}||_2 = 1$ that within the optimality tolerance $\epsilon$}
Initialize $n_0 = (\mathbf{l, u}) = \left( \{ 0 \}^p, \{ 1 \}^p\right)$ as the root node \;
Initialize node set $\mathcal{N} = \{ n_0 \}$ \;
Initialize $\hat{\mathbf{v}}$ and lower bound $lb$ using customized Rifle\;
Initialize upper bound $ub = \lambda_\text{max}(A, B)$\;
\While{$ub - lb > \epsilon$} {
    choose a node $(\mathbf{l}, \mathbf{u}) \in \mathcal{N}$\;
    choose a index $i$ where $l_i = 0, u_i = 1$\;
    \For{$val = 0, 1$}{
    $\text{newnode} = ((l_1, \dots, l_{i-1}, val, l_{i+1}, \dots, l_p), (u_1, \dots, u_{i-1}, val, u_{i+1}, \dots, u_p)); $ \\
    \uIf{Terminal(newnode)} {
        compute $\mathbf{v} = \argmax \mathbf{v}^{\intercal}A\mathbf{v}, \mathbf{v} \in \mathcal{V}(\mathbf{l_{new}, u_{new}}, k, B)$\;
        set $\text{upper bound} = \text{lower bound} = \mathbf{v}^{T}A\mathbf{v}$\;
        }
    \uElse{
        compute $\text{lower bound} = lower(newnode)$, with corresponding feasible point $\mathbf{v}$\;
        compute $\text{upper bound} = upper(newnode)$\;
        }
    \uIf{$\text{lower bound} > lb$}{
        set $lb = \text{lower bound}$\;
        set $\hat{\mathbf{v}} = \mathbf{v}$\;
        remove all nodes in $\mathcal{N}$ with $\text{upper bound} \leq lb$\;
        }
    \uIf{$\text{upper bound} > lb$}{
        add newnode to $\mathcal{N}$\;
        }
    }
    remove $(\mathbf{l, u})$ from $\mathcal{N}$\;
    update $ub$ to be the greatest value of upper bounds over $\mathcal{N}$\;
}
 
 \caption{Optimal SGEP}
 \label{algo_sgep}
\end{algorithm}
In step 3 of Algorithm \ref{algo_sgep}, we use a modified version of Truncated Rayleigh Flow Method (Rifle) proposed by \citet{tan2018sparse} to generate a good wart start, which is introduced in Section \ref{sec_upper_lower_bounds}. The heuristic methods to select nodes and dimensions in step 6 and 7 are discussed in Section \ref{node_selection_branch}. The convergence property of Optimal SGEP is established in Section \ref{convergence}.

\subsection{Upper and lower bounds}\label{sec_upper_lower_bounds}
Many bounds derived in this section require a projection step that takes any vector $\mathbf{v} \in \mathbb{R}^p$ as input and outputs the vector $\hat{\mathbf{v}}$ closest to $\mathbf{v}$ in the feasible set $\mathcal{V}(\mathbf{l, u}, k , B)$. That is $\hat{\mathbf{v}}$ is a solution of 
$$
\begin{aligned}
\min_{\hat{\mathbf{v}}} \quad & ||\mathbf{v} - \hat{\mathbf{v}}||_2^2 \\
\text{s.t. } &  \hat{\mathbf{v}} \in \mathcal{V}(\mathbf{l, u}, k , B)
\end{aligned}
\label{eq_trunc}
$$
The above optimization problem can be solved using Algorithm \ref{algo_trunc} which performs a simple truncation step on the original vector $\mathbf{v}$.

\begin{algorithm}[H]
\SetAlgoLined
\SetKwInOut{Input}{Input}
\SetKwInOut{Output}{Output} 
\Input{vector $\mathbf{v}$, cardinality $k$, lower and upper bounds $\mathbf{l, u}$ of the support, matrix $B$}
\Output{sparse vector $\hat{\mathbf{v}} \in \mathcal{V}(\mathbf{l, u}, k, B)$ that minimizes $||\mathbf{v - \hat{v}}||_2$ }
Compute the set $\mathcal{I}_1 = \{i \mid l_i = 1 \}$\;
Compute the set $\mathcal{I}_2$ by choosing the indices of the largest $k - |\mathcal{I}_1|$ absolute values of the vector $ \mathbf{v}_{\{i|l_i = 0, u_i = 1 \}}$\;
Compute $\hat{\mathbf{v}}$ by
$$
\hat{\mathbf{v}}_i = 
\begin{cases}
v_i & \text{if } i \in \mathcal{I}_1 \cup \mathcal{I}_2, \\
0 & \text{otherwise}.
\end{cases}
$$\\
Normalize 
$$\hat{\mathbf{v}} = \frac{\hat{\mathbf{v}}}{\hat{\mathbf{v}}^{T}B\hat{\mathbf{v}}}$$ \\
 \caption{$Truncate(\cdot, \mathbf{l, u}, k, B)$}
 \label{algo_trunc}
\end{algorithm}

\subsubsection{Upper bounds}
The first upper bound is derived based on the fact that if $\mathbf{v} \in \mathcal{V}(\mathbf{l, u}, k, B)$ then $\mathbf{v}$ is a feasible solution of 
$$
\begin{aligned}
\max_{\mathbf{v}} \quad & \mathbf{v}^{\intercal}A\mathbf{v} \\
\text{s.t. } \quad & \sum_{i=1}^p \lambda_i(\mathbf{s}_i^{\intercal}\mathbf{v})^2 = 1   \\
 & u_i = 0, \quad  \forall i \in \{i | u_i = 0 \}.
\end{aligned}
$$
The above optimization problem can reduce to a regular GEP:
$$
\begin{aligned}
\max_{\mathbf{v}} \quad & \mathbf{v}^{\intercal}A_{\mathbf{u}}\mathbf{v} \\
\text{s.t. } \quad & \mathbf{v}^{\intercal}B_{\mathbf{u}}\mathbf{v} = 1,
\end{aligned}
$$
where $A_{\mathbf{u}}, B_{\mathbf{u}}$ are matrices obtained by setting rows and columns of $A, B$ corresponding to $u_i = 0$ to zero. Thus, the first upper bound is 
\begin{equation}
    \text{upper bound 1} = \lambda_{\max}(A_\mathbf{u}, B_\mathbf{u}).
    \label{ub1}
\end{equation}

For the subproblem \eqref{sgep_pds}, we have $\mathbf{v}^{\intercal}A\mathbf{v} = \mathbf{v}^{\intercal}A_{\mathbf{w}}\mathbf{v}$ where $\mathbf{v} \in \mathcal{V}(\mathbf{l, u}, k, B)$, $\mathbf{w}$ is the support of $\mathbf{v}$. Since the trace of a semi-positive definite matrix is greater than or equal to the maximum eigenvalue, we have 
\begin{equation}
    \mathbf{v}^{\intercal}A_{\mathbf{w}}\mathbf{v} \leq \frac{\lambda_{\max}(A_\mathbf{w})}{\lambda_{\min}(B_\mathbf{w})} \leq \frac{tr(A_{\mathbf{w}})}{\lambda_{\min}(B_\mathbf{w})}.
    \label{ineq_trace}
\end{equation}

According the definition of $A_{\mathbf{w}}$, we obtain
$$
\begin{aligned}
& tr(A_{\mathbf{w}}) = \sum_{i=1}^pTruncate(diag(A), \mathbf{w, w}, k, B) \leq \sum_{i=1}^pTruncate(diag(A), \mathbf{l, u}, k, B), \\
& \lambda_{\min}(A_\mathbf{w}) = \min_{||\mathbf{v}||_2 = 1} \mathbf{v}^{\intercal}A_{\mathbf{w}}\mathbf{v} = \min_{||\mathbf{v}||_2 = 1, \text{supp}(\mathbf{v})=\mathbf{w}}\mathbf{v}^{\intercal}A\mathbf{v} \geq  \min_{||\mathbf{v}||_2 = 1}\mathbf{v}^{\intercal}A\mathbf{v} = \lambda_{\min}(A).
\end{aligned}
$$
Since $\mathbf{v}$ is arbitrary, the second upper bound is 
\begin{equation}
    \text{upper bound 2} = \frac{\sum_{i=1}^pTruncate(diag(A), \mathbf{l, u}, k, B)}{\lambda_{\min}(A)}.
    \label{ub2}
\end{equation}

According to \citet{scott1985accuracy}, Gershorin circle theorem provides an accurate bound of the largest eigenvalue of a real symmetric matrix. Using this theorem and the fact that $A_{ii} \geq 0$, we have
$$
\lambda_{\max}(A) \leq \max_{i}\sum_{j=1}^p|A_{ij}|.
$$
Then, utilize the above upper bound instead of the trace in the inequality \eqref{ineq_trace}, we obtain 
$$
\mathbf{v}^{\intercal}A\mathbf{v} = \mathbf{v}^{\intercal}A_{\mathbf{w}}\mathbf{v} \leq \frac{\max_{i}\sum_{j=1}^p|(A_{\mathbf{w}})_{ij}|}{\lambda_{\min}(B_{\mathbf{w}})}. 
$$
Taking the maximum value of both sides of the above inequality in the feasible set $\mathcal{V}(\mathbf{l, u}, k, B)$, the third upper bound can be derived as follows
$$
\begin{aligned}
\max_{\mathbf{v} \in \mathcal{V}(\mathbf{l, u}, k, B)} \mathbf{v}^{\intercal}A\mathbf{v} &\leq \frac{1}{\lambda_{\min}(B)} \max_{\mathbf{l} \leq \mathbf{w} \leq \mathbf{u}, \sum_{i}w_i=k}\left( \max_{i}\sum_{j=1}^p|(A_{\mathbf{w}})_{ij}| \right) \\
&= \frac{1}{\lambda_{\min}(B)} \max_{i}\left( \max_{\mathbf{l} \leq \mathbf{w} \leq \mathbf{u}, \sum_{i}w_i=k}\sum_{j=1}^p|(A_{\mathbf{w}})_{ij}| \right) \\
&= \frac{1}{\lambda_{\min}(B)} \max_{i, u_i=1}\left(\sum_{j=1}^p Truncate(|A_{i\cdot}|, \mathbf{l, u}, k, B)_j \right).
\end{aligned}
$$
Therefore, we obtain
\begin{equation}
    \text{upper bound 3} =\frac{1}{\lambda_{\min}(B)} \max_{i, u_i=1}\left(\sum_{j=1}^p Truncate(|A_{i\cdot}|, \mathbf{l, u}, k, B)_j \right).
    \label{ub3}
\end{equation}

Combining the three upper bounds \eqref{ub1}, \eqref{ub2}, \eqref{ub3}, the $upper$ function in \eqref{ub_lb_func} which compute a upper bound of a subproblem \eqref{sgep_pds} is defined as
\begin{equation}
    upper(\mathbf{l, u}, k, B) = \min (\text{upper bound 1}, \text{upper bound 2}, \text{upper bound 3}).
\end{equation}

\subsubsection{Lower bounds}

According to the definition of lower bounds in \eqref{ub_lb_func}, we only need to find a feasible vector $\mathbf{v}$ and a lower bound is given by $ \mathbf{v}^{\intercal}A\mathbf{v} $. The Rifle method proposed by \citet{tan2018sparse} can be exploited for a feasible vector. Note that the generalized Rayleigh quotient form of problem \eqref{sgep} is
\begin{equation}
    \max_{\mathbf{v}} \frac{\mathbf{v}^{\intercal}A\mathbf{v}}{\mathbf{v}^{\intercal}B\mathbf{v}}, \quad \text{s.t. } ||\mathbf{v}||_0 \leq k.
    \label{rayleigh}
\end{equation}
Rifle first computes the gradient of the objective function in \eqref{rayleigh}. Then the current vector is updated by its ascent direction, which ensures the output $\mathbf{v}_t$ is not worse than the initial vector. Next, Rifle truncates the updated vector to obtain sparsity. To adopt Rifle, we made two modifications to it.
Firstly, in the original algorithm, truncating a vector keeps elements with the largest $k$ absolute values and sets other elements to zero. Now, we enforce the truncated vector must contain entries where $l_i=1$ and not include entries where $u_i=0$, that is, we use our truncation Algorithm \ref{algo_trunc}. Secondly, the while loop are repeated $N_1$ times, where $N_1$ is a small number, not until convergence. The reason is explained in Section \ref{node_selection_branch}.

\begin{algorithm}[H]
 \label{algo_rifle}
\SetAlgoLined
\SetKwInOut{Input}{Input}
\SetKwInOut{Output}{Output} 
\Input{matrices $A, B$, cardinality $k$, step size $\eta$, initial vector $\mathbf{v}_0$, maximum number of iterations $N_1$}
\Output{$\mathbf{v}_t \in \mathcal{V}(\mathbf{l, u}, k, B)$}
$t=1$\;
 \While{$t \leq N_1$}{
    $\rho_{t-1} = \mathbf{v}_{t-1}^{\intercal}A\mathbf{v}_{t-1}/\mathbf{v}_{t-1}^{\intercal}B\mathbf{v}_{t-1}$ \;
    $C = I + (\eta/\rho_{t-1})\cdot (A - \rho_{t-1}B)$\;
    $\mathbf{v}_t^{'} = C\mathbf{v}_{t-1} / ||C\mathbf{v}_{t-1}||_2$\;
    $\mathbf{v}_t = Truncate(\mathbf{v}_t^{'}, \mathbf{l, u}, k, B)$\;
    $t = t + 1$\;
 }

 \caption{Customized Rifle}
\end{algorithm}

The Rifle requires that the initial vector $\mathbf{v}_0$ is close to the optimal solution of problem \eqref{sgep}. \citet{tan2018sparse} considered a convex formulation of \eqref{sgep} 
$$
\min_{P \in \mathbb{S}_+^{p}} -tr(AP) + \zeta ||P||_{1}, \quad \text{s.t. } ||B^{1/2}PB^{1/2}||_* \leq 1 \text{ and } ||B^{1/2}PB^{1/2}||_2 \leq 1,
$$
where $P$ represents the orthogonal projection onto the subspace spanned by the largest eigenvector. The nuclear norm forces the solution to have low rank and the spectral norm constrains the largest eigenvalue of the solution.
They solved it using linearized alternating direction method of multipliers (ADMM) (\citet{fang2015generalized}). In each iteration, the algorithm needs to perform a full singular value decomposition (SVD) of a $p \times p$ matrix, which is computationally expensive. It is natural that with more iterations, we can obtain a better initial vector. However, in Section \ref{node_selection_branch}, we'll show our branch and bound algorithm has little benefit from increasing the number of iterations.

\begin{algorithm}[H]
 \label{algo_LADMM}
\SetAlgoLined
\SetKwInOut{Input}{Input}
\SetKwInOut{Output}{Output} 
\Input{matrices $A, B$, tuning parameter $\zeta$, ADMM parameter $\nu$, maximum number of iterations $N_2$}
\Output{$P_t$}
Initialize with $P_0, H_0, \Gamma_0$\;
t = 1\;
 \While{$t \leq N_2$}{
     Update $P$ by solving the following LASSO problem:
    $$
    P_{t+1} = \argmin_{P}\frac{\nu}{2}||B^{1/2}PB^{1/2} - H_t + \Gamma_t||_F^2 - tr(AP) + \zeta ||P||_{1, 1};
    $$\\
    
    Let $\sum_{j=1}^p\omega_j\textbf{a}_j\textbf{a}_j^{\intercal}$ be the singular value decomposition of $\Gamma_t + B^{1/2}P_{t+1}B^{1/2}$ and let 
    $$
    \gamma^* = \argmin_{\gamma > 0}\gamma \quad 
    \text{s.t. } \sum_{j = 1}^p\min\{ 1, \max(\omega_j - \gamma, 0) \} \leq K;
    $$\\
    
    Update $H$ by 
    $$
    H_{t+1} = \sum_{j=1}^p\min\{ 1, \max(\omega_j - \gamma^*, 0) \}\mathbf{a}_j\mathbf{a}_j^{\intercal};
    $$\\
    
    Update $\Gamma$ by
    $$
    \Gamma_{t+1} = \Gamma_t + B^{1/2}P_{t+1}B^{1/2} - H_{t+1};
    $$\\
    t = t + 1\;
 }

 \caption{Linearized ADMM with early stop}
\end{algorithm}

Let $P_t$ be the output of Algorithm \ref{algo_LADMM}. Then we set the initial vector for Algorithm \ref{algo_rifle} to be the eigenvector corresponding to the largest eigenvalue. Tuning parameter $\zeta$ and ADMM parameter $\nu$ are set to be $\sqrt{\log (p) / n}$ and 1 respectively as suggested by \citet{tan2018sparse}. Let $\mathbf{v}_{rifle}$ be the output of Algorithm \ref{algo_rifle}. The first lower bound is 
\begin{equation}
     \text{lower bound 1} = \mathbf{v}_{rifle}^{\intercal}A\mathbf{v}_{rifle}
    \label{lb1}
\end{equation}

Another way of obtaining a feasible vector is truncating the eigenvector $\mathbf{v}_{\max}$ corresponding to the largest eigenvalue $\lambda_{\max}(A, B)$, that is 
\begin{equation}
    \text{lower bound 2} = (Truncate(\mathbf{v}_{\max}, \mathbf{l, u}, B))^{\intercal}A(Truncate(\mathbf{v}_{\max}, \mathbf{l, u}, B))
    \label{lb2}
\end{equation}

Combining the two lower bounds \eqref{lb1}, \eqref{lb2}, the $lower$ function in \eqref{ub_lb_func} which computes a lower bound of a subproblem \eqref{sgep_pds} is defined as
\begin{equation}
    lower(\mathbf{l, u}, k, B) = \max (\text{lower bound 1}, \text{lower bound 2}).
\end{equation}

\subsection{Node selection, branching and de-emphasizing lower bounds} \label{node_selection_branch}
In this section, we study some technical details of Optimal SGEP. First, we compare two strategies to select nodes: depth-first and best-first. Then we discuss several variable branching heuristics which may greatly affect the total number of explored nodes. Finally, we consider de-emphasizing lower bounds which is accomplished by limiting the number of iterations of Algorithm \ref{algo_rifle} and \ref{algo_LADMM}.  

We conducted numerical experiments to compare differences between various methods and parameter settings. We used model 3 in Section \ref{sec_desc_models} with $n = 200, p = 80$ to generate data and applied sliced inverse regression with 5 slices to obtain matrices $A$ and $B$. Then we ran Algorithm \ref{algo_sgep} on this problem with $k = 3$ one hundred times and documented averaged results. We report three important measures of the performance of Optimal SGEP: the time to obtain a feasible solution that will later be verified as a optimal solution (time to lower bound), the time for the algorithm to converge (time to upper bound), and the number of explored nodes. 

\subsubsection{Node selection}
The node selection strategies decides how we explore the enumeration tree. The two heuristics we consider are depth-first and best-first (\citet{lodi2017learning}). In depth-first, we select the most recently added node until it is terminal. Depth-first keeps the number of unexplored nodes small, but causes slow reduction of the bound gap. In best-fist, we choose the node with highest upper bound, which guarantees the total upper bound declines at each iteration. But algorithms using best-first need to maintain a large set of unexplored nodes since it continuously replace one node with its two child nodes. Although the two methods have different number of nodes need to maintain in the running process, they usually result in the same number of total explored nodes. 

The properties of depth-first and best-first motivate a two-phase method, which combines the two methods to balance the search objectives. In Algorithm \ref{algo_node_select}, we create a parameter called \emph{MaxDepth} which controls the maximum number of active nodes. If the number of nodes is smaller than MaxDepth, best-first is applied to quickly decrease the overall upper bound. Once the number exceeds MaxDepth, we employ depth-first to reduce the size of the set of active nodes until it falls below MaxDepth. 

As shown in Table \ref{table_node_select}, the drastic changes of MaxDepth lead to almost the same performance. However, we do not recommend setting MaxDepth too small since it causes slow reduction progress of overall upper bounds until the algorithm is close to convergence. The priority of depth-first and best-first depends on the specific problem. Therefore, we suggest dynamically adjusting this parameter according to different purposes.

\begin{algorithm}[H]
 \label{algo_node_select}
\SetAlgoLined
\SetKwInOut{Input}{Input}
\SetKwInOut{Output}{Output} 
\Input{MaxDepth, node set $\mathcal{N}$}
\Output{selected node}
$\text{n} = |\mathcal{N}|$\;
\eIf{$n < MaxDepth$}{
    set $\text{node} = \argmax_{\text{node} \in \mathcal{N}}upper(\text{node})$\;
}{
    set $\text{node} = \text{the most recently added node}$\;
}
 \caption{Node Selection}
\end{algorithm}

\begin{table}[!htbp] 
\centering
 \caption{Effects of increasing the maximum number of active nodes} 
 \label{table_node_select}
 \resizebox{\textwidth}{!}{
 \begin{threeparttable}  
 \begin{tabular}{cccc} 
  \toprule 
  MaxDepth  & Time(s) to lower bound & Time(s) to convergence & Number of explored nodes \\ 
  \midrule 
    10 &  0.77278 & 4.44018 & 342.96 \\
    20 &  0.70350 & 4.42414 & 342.85 \\
    40 &  0.70385 & 4.42484 & 342.85 \\
    80 &  0.70306 & 4.42404 & 342.85 \\
   200 &  0.70333 & 4.42729 & 342.85 \\
   400 &  0.70369 & 4.43190 & 342.85 \\
   800 &  0.70376 & 4.42211 & 342.85 \\
  1500 &  0.70511 & 4.42648 & 342.85 \\
  3000 &  0.70342 & 4.42644 & 342.85 \\
  6000 &  0.70205 & 4.42117 & 342.85 \\
 10000 &  0.70301 & 4.42290 & 342.85 \\
 20000 &  0.70371 & 4.42135 & 342.85 \\
  \bottomrule 
 \end{tabular} 
 \begin{tablenotes}    
        \footnotesize               
        \item[1] The worst performance occurs at $\text{MaxDepth} = 10$ and results under other settings have minor differences. 
      \end{tablenotes}            
\end{threeparttable} }
\end{table}

\subsubsection{Branching}
After selecting a node, we need to appropriately choose a variable to branch on, which corresponds to step 7 in Algorithm \ref{algo_sgep}. A good branching strategy can significantly reduce the number of explored nodes.

The first and most naive approach is randomly choosing a variable among the set $\mathcal{I} = \{ i | l_i=0, u_i=1 \}$. This method doesn't use any information from matrices $A$ and $B$. Therefore, it generally makes little progress at each step, which makes the algorithm extremely inefficient. The second method chooses a dimension $i$  that solves the problem
$$
\max_{i \in \mathcal{I}}\frac{\mathbf{e}_i^{\intercal}A\mathbf{e}_i}{\mathbf{e}_i^{\intercal}B\mathbf{e}_i}, 
$$
where $\mathbf{e}_i$ is the $i$-th standard basis. This method performs well in most cases, but fails when the spread of values $A_{ii}/B_{ii}$ is small. Another method computes $i$ by setting it to be the index of the largest absolute loading of the eigenvector corresponding to $\lambda_{\max}(A_{\mathcal{I}}, B_{\mathcal{I}})$. The drawback to it is the time consuming process of computing the eigenvector. The last heuristic selects $i = \argmax_{i \in \mathcal{I}}|(\mathbf{v}_{\max})_i|$, where $\mathbf{v}_{\max}$ is the largest eigenvector of $A$ and $B$.

The results in Table \ref{table_dim_select} show that selecting dimension according to $|(\mathbf{v}_{\max})_i|$ outperforms the others. Actually, all three methods except random selection have similar performance. Thus, in practice, we suggest setting the fourth method as default and the second and third methods as alternatives.

\begin{table}[!htbp] 
\centering
 \caption{ Comparison of four dimension selection methods}
\label{table_dim_select}
 \resizebox{\textwidth}{!}{
 \begin{threeparttable}  
 \begin{tabular}{cccc} 
  \toprule 
  Method  & Time(s) to lower bound & Time(s) to convergence & Number of explored nodes \\ 
  \midrule 
    random            &  93.279  &  137.815  &  13868.6 \\
    $A_{ii} / B_{ii}$ &  0.68025 &  5.88354  &  494.79 \\
    $|\mathbf{v}_{\max}(A_{\mathcal{I}}, B_{\mathcal{I}})_i|$ &  0.69008 &   4.97604  &  400.73 \\
   $|(\mathbf{v}_{\max})_i|$&  0.70652  &  4.38796  &  342.85 \\
  \bottomrule 
 \end{tabular} 
 \begin{tablenotes}    
        \footnotesize               
        \item[1] The worst performance occurs when using random selection and the fourth method outperforms others.
      \end{tablenotes}            
\end{threeparttable} }

\end{table}

\subsubsection{De-emphasizing lower bound}
The goal of using Rifle is to obtain a lower bound at a node. The parameters $N_1$ and $N_2$ which control the number of iterations of Algorithm \ref{algo_rifle} and \ref{algo_LADMM} deeply affect the quality of obtained lower bounds. Clearly, taking more iterations can improve the quality of lower bounds. However, it will increase the work per node at the same time, which may make the overall performance of Optimal SGEP worse. 

We see from Table \ref{table_effects_algo_rifle} that as $N_1$ increases, the time to lower bounds may decrease when $N_1$ takes some values, but the time for the algorithm to converge increase continually. Since each iteration of Algorithm \ref{algo_LADMM} requires a full singular value decomposition, spending more time on obtaining a good initial vector can seriously hinder the performance of the branch and bound algorithm as shown in Table \ref{table_effects_LADMM}.

Therefore, experiment studies suggest we should choose smaller $N_1$ and $N_2$. In practice, two steps of Customized Rifle combining with one step of ADMM can achieve a good lower bound.

\begin{table}[!htbp] 
\centering
 \caption{ Effects of increasing the number of iterations of Algorithm \ref{algo_rifle} } 
 \label{table_effects_algo_rifle}
 \resizebox{\textwidth}{!}{
 \begin{threeparttable}  
 \begin{tabular}{cccc} 
  \toprule 
  $N_1$ & Time(s) to lower bound & Time(s) to convergence & Number of explored nodes \\ 
  \midrule 
  1 & 0.70368 & 4.42460 & 342.85 \\
  2 & 0.70727 & 4.46465 & 342.85 \\
  3 & 0.70721 & 4.49915 & 342.85 \\
  4 & 0.70800 & 4.54188 & 342.85 \\
  5 & 0.70849 & 4.57578 & 342.85 \\
  6 & 0.71090 & 4.61840 & 342.85 \\
  7 & 0.71265 & 4.66210 & 342.85 \\
  8 & 0.71463 & 4.70206 & 342.85 \\
  9 & 0.71448 & 4.74727 & 342.85 \\
 10 & 0.71176 & 4.78149 & 342.85 \\
 20 & 0.71552 & 4.83138 & 342.85 \\
 30 & 0.71798 & 4.87349 & 342.85 \\
 40 & 0.71893 & 4.92148 & 342.85 \\
 50 & 0.71871 & 4.96769 & 342.85 \\
  \bottomrule 
 \end{tabular} 
 \begin{tablenotes}    
        \footnotesize               
        \item[1] Branch and bound algorithm benefits a little from increasing the number of iterations of obtaining lower bounds.     
      \end{tablenotes}            
\end{threeparttable} }
\end{table}

\begin{table}[!htbp] 
\centering
 \caption{ Effects of increasing the number of iterations of Algorithm \ref{algo_LADMM} } 
 \label{table_effects_LADMM}
 \resizebox{\textwidth}{!}{
 \begin{threeparttable}  
 \begin{tabular}{cccc} 
  \toprule 
  $N_2$ & Time(s) to lower bound & Time(s) to convergence & Number of explored nodes \\ 
  \midrule 
  1 & 0.71643 &  4.47632 & 342.85 \\
  2 & 0.75290 &  6.68824 & 342.85 \\
  3 & 0.79193 &  8.85753 & 342.85 \\
  4 & 0.83765 & 11.0351  & 342.85 \\
  5 & 0.88091 & 13.1915  & 342.85 \\
  6 & 0.92437 & 15.3506  & 342.85 \\
  7 & 0.96785 & 17.4755  & 342.85 \\
  8 & 1.01115 & 19.6039  & 342.85 \\
  9 & 1.04986 & 21.7109  & 342.85 \\
 10 & 1.09356 & 23.8220  & 342.85 \\
 20 & 1.13682 & 25.9278  & 342.85 \\
 30 & 1.17803 & 28.0471  & 342.85 \\
 40 & 1.21925 & 30.1605  & 342.85 \\
 50 & 1.26143 & 32.2851  & 342.85 \\
  \bottomrule 
 \end{tabular} 
  \begin{tablenotes}    
        \footnotesize               
        \item[1] More time spent computing the initial vector results in a drastic increase in run time.    
      \end{tablenotes}            
\end{threeparttable} }
\end{table}

\subsection{Convergence} \label{convergence}
In this section, we present the convergence property of our algorithm. The proof is straightforward and a result of the general branch and bound algorithm.

\begin{definition}
An $\epsilon$-optimal solution of problem \eqref{sgep_mio} is a vector $\hat{\mathbf{v}}$ that satisfies
$$
\hat{\mathbf{v}}^{\intercal}A\hat{\mathbf{v}} \geq \mathbf{v}^{\intercal}A\mathbf{v} - \epsilon, \quad \forall \mathbf{v} \in \mathcal{V}(\{ 0 \}^p, \{ 1 \}^p, k, B).
$$
\end{definition}

\begin{theorem}\label{theo_conv}
The algorithm, Optimal SGEP, produces an $\epsilon$-optimal solution of the sparse generalized eigenvalue problem \eqref{sgep_mio} within finite iterations.
\end{theorem}
\begin{proof}
The algorithm seeks a node in the enumeration tree, so the number of iterations is bounded by $2^{p+1}-1$, the total number of nodes in the tree. Nodes are pruned only if their feasible sets are empty, or they don't have feasible solutions that are better than the current best solution. Therefore, we will not exclude any optimal solution before termination. Once the gap between overall lower and upper bounds falls below $\epsilon$, we obtain a $\epsilon$-optimal solution of problem \eqref{sgep_mio}. 
\end{proof}

Theorem \ref{theo_conv} shows that Optimal SGEP is an exact method for solving non-convex sparse generalized eigenvalue problems. In particular, we can compute the globally optimal solution of SGEP by setting $\epsilon$ to 0.

\subsection{Tuning parameter selection}\label{sec_tuning}
Algorithm \ref{algo_sgep} involves one tuning parameter $k$, which controls the sparsity of the estimated central subspace. We recommend using the following criterion to choose $k$
\begin{equation}
    \sum_{i=1}^p ||B^{-1}\sqrt{A}_{\cdot j} - \hat{\mathbf{v}}_k\hat{\mathbf{v}}_k^{\intercal}\sqrt{A}_{\cdot j}||_B^2 + \gamma \cdot df,
    \label{eq_aic_bic}
\end{equation}
where $\hat{\mathbf{v}}_k$ denotes the output of Algorithm \ref{algo_sgep} given $k$. Following the suggestion of \citet{chen2010coordinate}, we set $\gamma = 2 / n$ and $\gamma = \log(n) / n$ for AIC-type (\cite{akaike1998information}) and BIC-type (\citet{schwarz1978estimating}) criteria respectively. $\gamma$ represents the effective number of parameters and is estimated by the number of nonzero elements of $\hat{\mathbf{v}}_k$. According the discussion of \citet{li2007sparse}, we select the $k$ that minimize equation \eqref{eq_aic_bic}.

\section{Simulation Studies}	 \label{simulation_section}
In this section, we first compare our branch and bound algorithm to the state-of-art Rifle method (\citet{tan2018sparse}) in terms of solution qualities. Then, we show the proposed algorithm with early stop has a good performance on large datasets.

All experiments in this section were performed on an eight-core 3.10 GHz processor (Intel i9-9900) with 16GB RAM. Our branch and bound method was implemented in Julia while the Rifle method was run with the authors' published R package (\citet{rifle_Rpackage}).

\subsection{Description of Models}\label{sec_desc_models}
Datasets used in Section \ref{comp_qualities_subsec} and \ref{comp_time_subsec} were generated from the following four models. For each model, we sampled $\mathbf{x}$ from a multivariate normal distribution $\mathcal{N}(0, \Sigma)$, where $\Sigma_{ij} = 0.5^{|i-j|}$ for $1 \leq i, j \leq p$, $\epsilon$ from the standard normal distribution and $\mathbf{x}$ and $\epsilon$ are independent.
\begin{enumerate}
    \item A linear regression model with large signal-to-noise ratio:
    $$
    y = x_1 + x_2 + x_3 + 0.5 \epsilon. 
    $$
    In this model, the central subspace is spanned by $\boldsymbol{\beta}_1 = (1, 1, 1, 0, \dots, 0)^{\intercal}$ with $p-3$ zeros. 
    \item A linear regression model with small signal-to-noise ratio:
    $$
    y = x_1 + x_2 + x_3 + 2 \epsilon.
    $$
    Same as the first model, the central subspace is spanned by $\boldsymbol{\beta}_1 = (1, 1, 1, 0, \dots, 0)^{\intercal}$ with $p-3$ zeros. But the error was four times the first model. 
    \item A non-linear regression model:
    $$
    y = 1 + \exp\{ (x_1 + x_2 + x_3 / \sqrt{3} \} + \epsilon.
    $$
    In this model, the central subspace is spanned by $\boldsymbol{\beta}_1 = (1, 1, 1, 0, \dots, 0)^{\intercal}$ with $p-3$ zeros.
    \item A non-linear regression model:
    $$
    y = sign(x_1 + x_2 + x_3 + x_4) \cdot \log (|x_{p-3} + x_{p-2} + x_{p-1} + 5|) + 0.1\epsilon,
    $$
    where $sign$ is a function that extracts the sign of a real number. In this model, the central subspace is spanned by the directions $\boldsymbol{\beta}_1 = (1, 1, 1, 1, 0, \dots, 0)^{\intercal}$ and $\boldsymbol{\beta}_2 = (0, \dots, 0, 1, 1, 1, 0)^{\intercal}$.
\end{enumerate}

\subsection{Comparison of solution qualities} \label{comp_qualities_subsec}
To evaluate solution qualities of different algorithms, we'll use three summary statistics$-$TPR, FPR, and $\Delta$. The true positive rate TPR is defined as the ratio of the number of correctly identified relevant variables to the number of truly relevant variables, and the false positive rate FPR is defined as the ratio of the number of falsely identified relevant variables to the number of truly irrelevant variables. $\Delta$ represents the distance between the estimated central subspace $\hat{\boldsymbol{\beta}}$ and the true central subspace $\boldsymbol{\beta}$ and is defined as 
\begin{equation}
    \Delta(\boldsymbol{\beta}, \hat{\boldsymbol{\beta}}) = ||P_{\boldsymbol{\beta}} - P_{\hat{\boldsymbol{\beta}}}||_F,
    \label{eq_distance_subspaces}
\end{equation}
where $P_{\boldsymbol{\beta}}, P_{\hat{\boldsymbol{\beta}}}$ are projection matrices onto subspaces $\boldsymbol{\beta}$ and $\hat{\boldsymbol{\beta}}$, respectively. 

We ran simulations on all four models with problem sizes $(n, p) = (150, 50)$ and $(n, p) = (300, 80)$. We used SIR with five slices and PFC with $f(y) = (|y|, y^2, y^3)^{\intercal}$ to generate matrices $A$ and $B$ for Optimal SGEP and Rifle method. For these methods, we denoted BB-SIR, BB-PFC, Rifle-SIR, and Rifle-PFC. The BIC criterion was used to select the tuning parameter $k$ for all methods. To assess the quality of estimated central subspaces, we computed means (standard errors) of true positive rate TPR, false positive rate FPR, and distance $\Delta$ over 100 datasets. 

The simulation results of four models are summarized in Tables \ref{table_simulation_model_1}-\ref{table_simulation_model_4}, respectively. 
These tables show that the proposed branch and bound algorithm outperforms Rifle in terms of TPR and $\Delta$ under all settings. Both methods have extremely small FPR, but Rifle performs slightly better in some situation. It should be noticed that the superiority of BB becomes more significant when the ratio of $n$ to $p$ get larger. In particular, BB correctly identified all active and inactive predictors of model 1 and 3 when $n = 300, p = 80$.

\begin{table}[!htbp]	
	\centering
	\caption{Summary of Model 1}
	\label{table_simulation_model_1}
	\resizebox{\textwidth}{!}{
		\begin{threeparttable}		
			\begin{tabular}{cccccc}
				\toprule
				\multirow{2}{*}{Problem Size}&\multirow{2}{*}{Measure}&\multicolumn{4}{c}{  Method }\cr			
				\cmidrule(lr){3-6}
				& & BB-SIR & BB-PFC & Rifle-SIR & Rifle-PFC \cr
				\midrule
$n=150, p=50$ & TPR & 0.997(0.03) & 0.987(0.07) & 0.673(0.05) & 0.670(0.03) \cr
              & FPR & 0.000(0.00) & 0.001(0.00) &0.000(0.00)&0.000(0.00)\cr	
&$\Delta$& 0.113(0.09) & 0.245(0.17) & 0.813(0.09) & 0.822(0.05) \cr

$n=300, p=80$ & TPR & 1.000(0.00) & 1.000(0.00) & 0.667(0.00) & 0.667(0.00) \cr
       & FPR & 0.000(0.00) & 0.000(0.00) & 0.000(0.00) & 0.000(0.00)\cr	
&$\Delta$& 0.081(0.04) & 0.166(0.09) & 0.823(0.01) & 0.825(0.01) \cr
				\bottomrule
			\end{tabular}	
	\end{threeparttable}}
\end{table}

\begin{table}[!htbp]	
	\centering
	\caption{Summary of Model 2}
	\label{table_simulation_model_2}
	\resizebox{\textwidth}{!}{
		\begin{threeparttable}		
			\begin{tabular}{cccccc}
				\toprule
				\multirow{2}{*}{Sample Size}&\multirow{2}{*}{Measure}&\multicolumn{4}{c}{  Method }\cr			
				\cmidrule(lr){3-6}
				& & BB-SIR & BB-PFC & Rifle-SIR & Rifle-PFC \cr
				\midrule
$n=150, p=50$ & TPR & 0.870(0.16) & 0.790(0.18) & 0.663(0.07) & 0.657(0.06) \cr
       & FPR & 0.000(0.00) & 0.013(0.01) & 0.001(0.01) & 0.002(0.01)\cr	
&$\Delta$& 0.440(0.33) & 0.673(0.34) & 0.834(0.08) & 0.874(0.08) \cr

$n=300, p=80$ & TPR & 0.997(0.03) & 0.947(0.12) & 0.667(0.00) & 0.667(0.00) \cr
       & FPR & 0.000(0.00) & 0.004(0.01) & 0.000(0.00) & 0.000(0.00)\cr	
&$\Delta$& 0.190(0.12) & 0.396(0.26) & 0.828(0.01) & 0.851(0.05) \cr
				\bottomrule
			\end{tabular}	
	\end{threeparttable}}
\end{table}

\begin{table}[!htbp]	
	\centering
	\caption{Summary of Model 3}
	\label{table_simulation_model_3}
	\resizebox{\textwidth}{!}{
		\begin{threeparttable}		
			\begin{tabular}{cccccc}
				\toprule
				\multirow{2}{*}{Sample Size}&\multirow{2}{*}{Measure}&\multicolumn{4}{c}{  Method }\cr			
				\cmidrule(lr){3-6}
				& & BB-SIR & BB-PFC & Rifle-SIR & Rifle-PFC \cr
				\midrule
$n=150, p=50$ & TPR & 0.937(0.13) & 0.997(0.03) & 0.670(0.03) & 0.670(0.03) \cr
       & FPR & 0.000(0.00) & 0.000(0.00) & 0.000(0.00) & 0.000(0.00)\cr	
&$\Delta$& 0.320(0.27) & 0.213(0.13) & 0.825(0.08) & 0.823(0.07) \cr

$n=300, p=80$ & TPR & 1.000(0.00) & 1.000(0.00) & 0.667(0.00) & 0.667(0.00) \cr
       & FPR & 0.000(0.00) & 0.000(0.00) & 0.000(0.00) & 0.000(0.00)\cr	
&$\Delta$& 0.173(0.09) & 0.149(0.08) & 0.826(0.01) & 0.826(0.01) \cr
				\bottomrule
			\end{tabular}	
	\end{threeparttable}}
\end{table}

\begin{table}[!htbp]	
	\centering
	\caption{Summary of Model 4}
	\label{table_simulation_model_4}
	\resizebox{\textwidth}{!}{
		\begin{threeparttable}		
			\begin{tabular}{cccccc}
				\toprule
				\multirow{2}{*}{Sample Size}&\multirow{2}{*}{Measure}&\multicolumn{4}{c}{  Method }\cr			
				\cmidrule(lr){3-6}
				& & BB-SIR & BB-PFC & Rifle-SIR & Rifle-PFC \cr
				\midrule
$n=150, p=50$ & TPR & 0.831(0.07) & 0.779(0.08) & 0.724(0.04) & 0.715(0.05) \cr
       & FPR & 0.005(0.01) & 0.011(0.01) & 0.004(0.01) & 0.004(0.01)\cr	
&$\Delta$& 0.839(0.17) & 0.933(0.19) & 1.050(0.06) & 1.062(0.09) \cr

$n=300, p=80$ & TPR & 0.937(0.07) & 0.876(0.08) & 0.719(0.02) & 0.714(0.03) \cr
       & FPR & 0.000(0.00) & 0.003(0.01) & 0.000(0.00) & 0.000(0.00)\cr	
&$\Delta$& 0.480(0.28) & 0.680(0.27) & 1.019(0.05) & 1.035(0.06) \cr

				\bottomrule
			\end{tabular}	
	\end{threeparttable}}
\end{table}

\subsection{Scaling Optimal-SGEP to large datasets} \label{comp_time_subsec}
To test the performance of branch and bound algorithm on large datasets, we generated samples of various sizes using model 1, 2 and 4. Since the objective of these experiments is to show Optimal-SGEP is scalable to big datasets, we assume that $k$ is selected correctly. If problem sizes are too large such that Optimal-SGEP cannot converge within $120s$, we compute TPR, FPR and $\Delta$ to assess the quality of solutions obtained using two time caps: 120s and 300s if $p < 1000$; 180s and 360s if $p = 1000$. If the algorithm converges within 120s, in addition to calculating TPR, FPR and $\Delta$, we also record time to lower bound and time to upper bound.

The results averaged over 20 datasets are reported in Table \ref{table_scaling_converge} and \ref{table_scaling_not_converge}. From Table \ref{table_scaling_converge}, we can observe that for problems with simple structures, such as model 1, Optimal-SGEP quickly attains the optimal solution and proves it is optimal. If the problems are complicated, such as model 3 and model 4, Optimal-SGEP can still quickly identify the optimal solution as a feasible solution, although it spends increasingly more time bridging the gap between lower and upper bounds as $p$ getting larger. This property suggests we can adopt an early stop strategy when applying Optimal-SGEP to large datasets. Table \ref{table_scaling_not_converge} shows the performance of Optimal-SGEP with different time caps on large datasets. It is notable that our algorithm can provide high-quality solutions for problems with $p=1000s$ in minutes. We also found doubling the running time of Optimal-SGEP hardly improves the quality of solutions.

Therefore, we believe these experiments prove that Optimal-SGEP is appropriate for both small and large problems and gives reliable solution within minutes.

\begin{table}[!htbp]	
	\centering
	\caption{Time to lower bound (TimeLB) and time to upper bound (TimeUB) are reported if the algorithm converges in 120s.}
	\label{table_scaling_converge}
	\resizebox{\textwidth}{!}{
		\begin{threeparttable}		
			\begin{tabular}{ccccccc}
				\toprule
				\multirow{2}{*}{Problem Size}&\multirow{2}{*}{Model}&\multicolumn{5}{c}{Measure}\cr	
				\cmidrule(lr){3-7}
				& & TimeLB & TimeUB & TPR & FPR & $\Delta$ \cr
				\midrule
$n=200,p=50$ & Model 1 & 0.086 & 0.092 & 1.000 & 0.000 & 0.076  \cr
             & Model 3 & 0.086 & 0.380 & 1.000 & 0.000 & 0.180  \cr
             & Model 4 & 0.461 & 2.076 & 0.864 & 0.020 & 0.797  \cr
$n=400,p=100$& Model 1 & 0.341 & 0.360 & 1.000 & 0.000 & 0.069  \cr
             & Model 3 & 0.340 & 2.185 & 1.000 & 0.000 & 0.130  \cr
             & Model 4 & 2.258 & 17.216 & 0.871 & 0.009 & 0.777  \cr
$n=1000,p=250$& Model 1 & 2.870 & 3.017 & 1.000 & 0.000 & 0.042  \cr
              & Model 3 & 2.877 & 39.743 & 1.000 & 0.000 & 0.078  \cr 
$n=2000,p=500$& Model 1 & 20.043 & 21.047 & 1.000 & 0.000 & 0.031  \cr

				\bottomrule
			\end{tabular}	
	\end{threeparttable}}
\end{table}

\begin{table}[!htbp]	
	\centering
	\caption{True and false positive rate, and distance $\Delta$ are computed within different time caps. For sample size $n = 4000, p=1000$, we set 180s and 360s as time caps. For other sample sizes, we set 120s and 300s as time caps.}
	\label{table_scaling_not_converge}
	\resizebox{\textwidth}{!}{
		\begin{threeparttable}		
			\begin{tabular}{cccc}
				\toprule
				\multirow{2}{*}{Problem Size}&\multirow{2}{*}{Model}&\multicolumn{2}{c}{(TPR, FPR, $\Delta)$}\cr	
				\cmidrule(lr){3-4}
				& & Time Cap $=120/180$ & Time Cap $=300/360$ \cr
				\midrule
$n=1000,p=250$& Model 4& (0.864, 0.004, 0.758) & (0.864, 0.003, 0.751)  \cr
$n=2000,p=500$& Model 3& (1.000, 0.000, 0.059) & (1.000, 0.000, 0.058)  \cr
              & Model 4& (0.846, 0.002, 0.743) & (0.864, 0.002, 0.741)   \cr 
$n=4000,p=1000$& Model 1& (1.000, 0.000, 0.020) & (1.000, 0.000, 0.019)  \cr
               & Model 3& (1.000, 0.000, 0.042) & (1.000, 0.000, 0.045)   \cr & Model 4& (0.850, 0.000, 0.794) & (0.879, 0.001, 0.737)   \cr
				\bottomrule
			\end{tabular}	
	\end{threeparttable}}
\end{table}

\section{Conclusion}\label{sec_conclusion}
In this paper, we propose a mixed integer optimization algorithm for sparse SDR problems using the fact that many SDR methods are special cases of generalized eigenvalue problems. The algorithm produces certifiably optimal solution and directly controls the sparsity. In particular, the simulation results show Optimal SGEP is computationally tractable for large problems. We believe our work demonstrates that discrete optimization algorithm is applicable to problems in Statistics.

Our proposed algorithm works well under the assumption that $n > p$. In the high dimensional settings in which $p > n$, matrices $A$ and $B$ are semi-definite, which implies $\lambda_{\max}(A, B)$ could be infinity. Therefore, the most important ingredient of branch and bound algorithm will fail. We must propose another formulation or discrete algorithm to handle the high dimensional cases. 

\newpage
\renewcommand\refname{References}

\bibliographystyle{apalike}
\bibliography{References}

\end{document}